\documentclass[12pt]{article}
\usepackage{amsmath,amsthm,amssymb}
\usepackage{graphicx,subfig,url,verbatim}
\usepackage[pdftex]{hyperref}
\hypersetup{
 colorlinks=true,
 urlcolor=blue
}

\newtheorem{theorem}{Theorem} 
\newtheorem{proposition}[theorem]{Proposition} 

\newtheorem{lemma}[theorem]{Lemma}

\theoremstyle{definition}

\theoremstyle{remark}

\author{Dashiell E.A.\,Fryer \\ The University of Illinois at Urbana-Champaign}
\title{The Uniform Distribution in Incentive Dynamics}
\date{\today}

\begin{document}
\maketitle

\begin{abstract}
The uniform distribution is an important counterexample in game theory as many of the canonical game dynamics have been shown not to converge to the equilibrium in certain cases. In particular none of the canonical game dynamics converge to the uniform distribution in a form of rock-paper-scissors where the amount an agent can lose is more than the agent can win, despite this fact, it is the unique Nash equilibrium. I will show that certain incentive dynamics are asymptotically stable at the uniform distribution when it is an incentive equilibrium.
\end{abstract}

Incentive dynamics~\cite{fryer2012existence} are given by \[ \dot{x}_{i\alpha} = \varphi_{i\alpha}(x) - x_{i\alpha}\sum_\beta\varphi_{i\beta}(x)\] where $\varphi(x)$ is the a valid incentive for the game. It was shown that if the incentive for a finite game is continuous, there exists a fixed point characterized by \[\varphi_{i\alpha}(\hat{x})=\hat{x}_{i\alpha}\sum_\beta\varphi_{i\beta}(\hat{x})\ \forall\alpha,i.\] Notice that if this occurs at the uniform distribution, either $\varphi_{i\alpha}(\hat{x})$ are all zero, or they are all the same for each agent.

Nash's original incentive function is fixed if and only if all the component incentives are zero and thus it can only be in the first case described above. In contrast, the incentive function given by $\varphi^D_{i\alpha}(x) = \sum_\gamma(a_{\alpha\gamma}-u_i(x))_+$ is only zero when $u_i(x)\geq \max_\gamma u_i(e_\alpha,e_\gamma)$, where $e_\gamma\in S_{-i}$ which can occur at the uniform distribution only if the game is constant, which is a degenerate case of little interest. Despite their differences we will demonstrate that the two can agree under certain circumstances. Also, we will see that the latter incentive is globally asymptotically stable at a uniform Nash equilibrium where the canonical dynamics fail to converge.

\subsection{A Bad Game of Rock-Paper-Scissors}

The standard game of Rock-Paper-Scissors (RPS) is given as a two person zero sum game with payoffs given in the table below on the left.
\[ \begin{array}[c]{cc}
\begin{array}{|rr|rr|rr|}
\hline 0,&0 & -1,&1 & 1,&-1 \\
\hline 1,&-1 & 0,&0 & -1,&1 \\
\hline -1,&1 & 1,&-1 & 0,&0 \\
\hline \end{array} & \begin{array}{|rr|rr|rr|}
\hline 0,&0 & -b,&a & a,&-b \\ 
\hline a,&-b & 0,&0 & -b,&a \\
\hline -b,&a & a,&-b & 0,&0 \\
\hline \end{array} 
\end{array}\]
To the right of the RPS payoffs we have a generalized RPS with $a$ and $b$ both positive. The case when $b>a$, or an agent can lose more than it can win, is an important example of a game. The unique Nash equilibrium for this game is the uniform distribution. We have seen many examples of incentive dynamics that have Nash equilibrium as their interior fixed points, such as the replicator equations, projection dynamics, the logit equations, best reply dynamics, and the Brown-von Neumann-Nash equations. However, in every one of these cases the dynamics do not converge to the unique equilibrium as shown in the figures below\footnote{The images were produced using the Dynamo Mathematica package developed by Sandholm, Dokumaci, and Franchetti~\cite{website:sandholm2011dynamo}. Colors indicate speed: blue is slowest and red is fastest}. This leads us to the natural question: does any incentive dynamic converge to a rest point from any initial point?

\begin{figure}[ht]
\centering
\subfloat[BNN]{\label{fig:BNN}\includegraphics[width=0.5\textwidth]{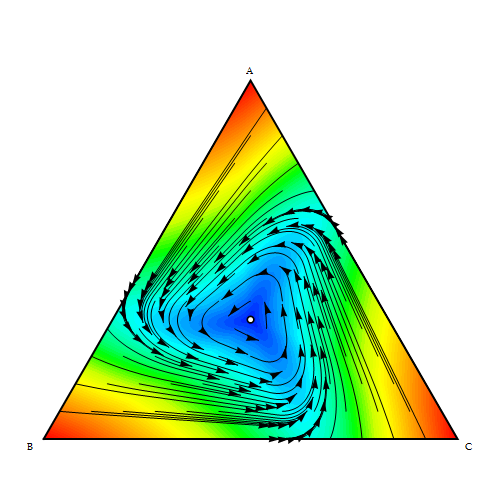}}
\subfloat[Logit (0.2)]{\label{fig:logit}\includegraphics[width=0.5\textwidth]{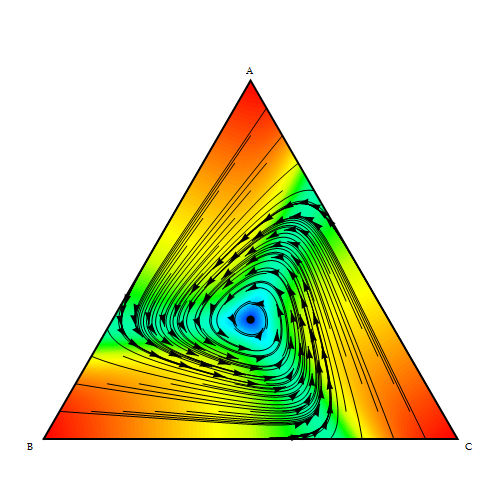}}\\
\subfloat[Smith]{\label{fig:Smith}\includegraphics[width=0.5\textwidth]{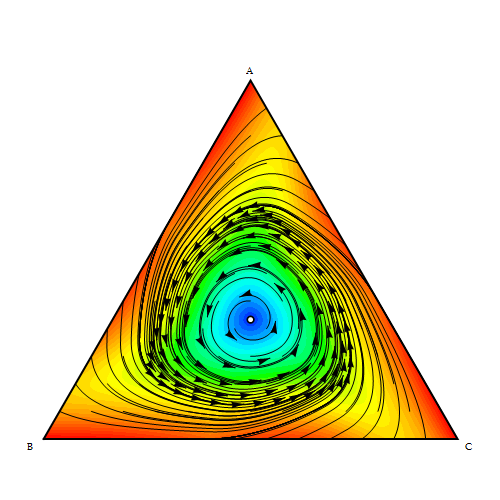}}
\caption{Stable limit cycles in Bad RPS}
\end{figure}

\begin{figure}[hb]
\centering
\subfloat[standard RPS]{\label{fig:replicator_RPS}\includegraphics[width=0.5\textwidth]{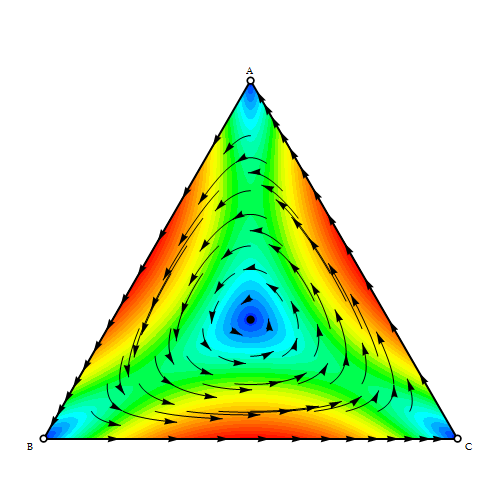}}
\subfloat[bad RPS]{\label{fig:replicator_BadRPS}\includegraphics[width=0.5\textwidth]{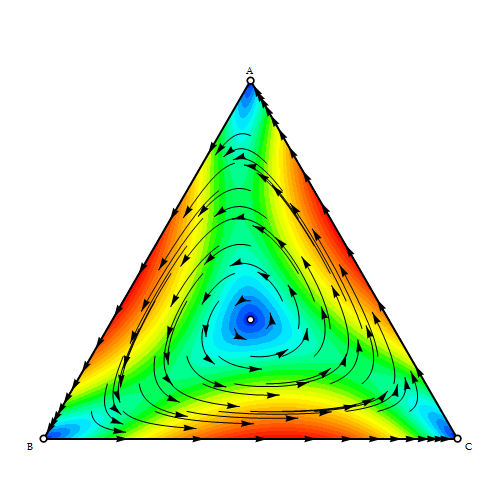}}
\caption{The replicator dynamics display invariant limit cycles and unstable equilibrium points in the RPS}
\end{figure}

\begin{figure}[ht]
\centering
\subfloat[standard RPS]{\label{fig:projection_RPS}\includegraphics[width=0.5\textwidth]{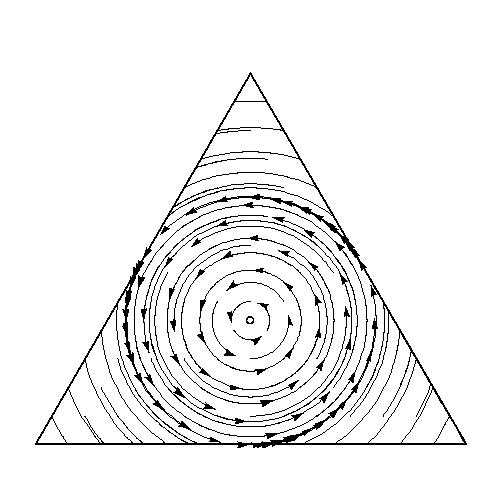}}
\subfloat[bad RPS]{\label{fig:projection_BadRPS}\includegraphics[width=0.5\textwidth]{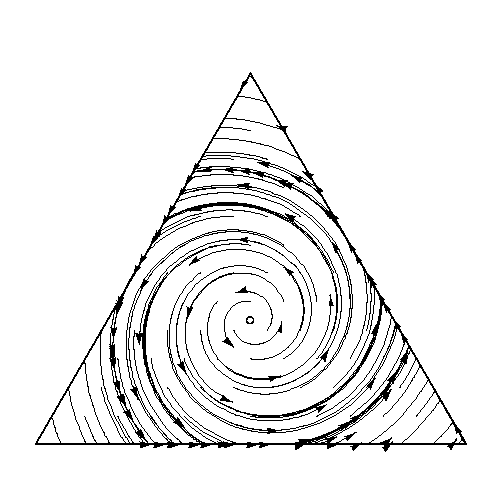}}
\caption{The projection dynamics display invariant limit cycles and unstable equilibrium points in the RPS}
\end{figure}

\begin{figure}[ht]
\centering
\subfloat[bad RPS]{\label{fig:Dash_BadRPS}\includegraphics[width=0.5\textwidth]{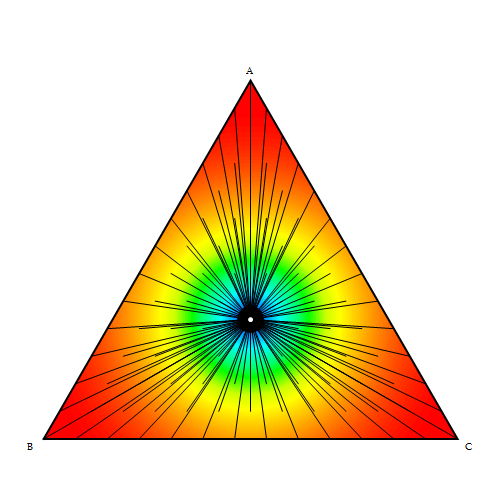}}
\subfloat[RPS]{\label{fig:Dash_RPS}\includegraphics[width=0.5\textwidth]{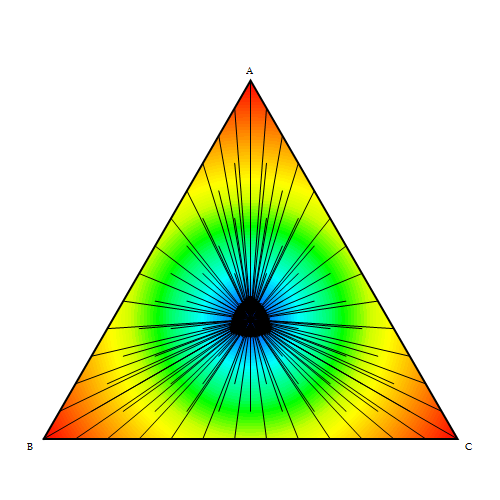}}\\
\subfloat[good RPS]{\label{fig:Dash_GoodRPS}\includegraphics[width=0.5\textwidth]{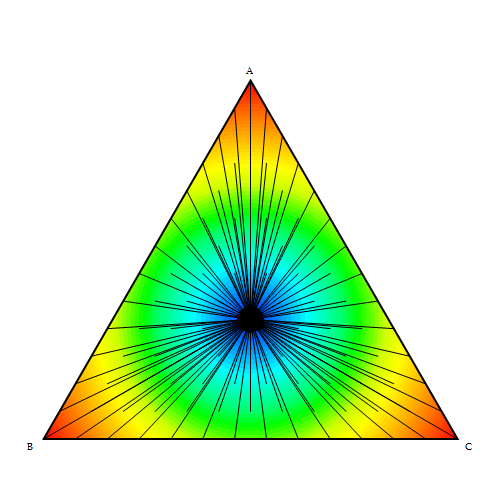}}
\caption{Global asymptotic stability of the uniform distribution in the simultaneous updating dynamics}
\end{figure}

\section{Agreement Among Incentives}

We note the incentive $\varphi^D_{i\alpha}(x) = \sum_\gamma(a_{\alpha\gamma}-u_i(x))_+$ can be rewritten in the form $\varphi^D_{i\alpha}(x) = \sum_\gamma(u_i(e_\alpha,x_{-i})-u_i(x)+a_{\alpha\gamma}-u_i(e_\alpha,x_{-i}))_+$, which shows that it is similar to a Nash comparison in that we are checking the payoff given the other agents' strategies are fixed. However, we are tempering that comparison by taking away the amount by which the agent is not receiving a preferred payoff available in the game. We will now show there is a class of games, which includes general RPS, with the property that the uniform distribution is a Nash equilibrium as well as an incentive equilibrium for $\varphi^D(x)$. First we will need the following lemma.
\begin{lemma}
If $A_i$ is the payoff matrix for the $i$th agent, then $\hat{x}$ is a Nash equilibrium where each agent is using the uniform distribution over its strategies if and only if for each $i$, $A_i$ has an equal sum across rows.
\end{lemma}
\begin{proof}
We begin by noting that for an interior Nash equilibrium we must have \[u_i(e_1,x_{-i}) = \ldots = u_i(e_{s_i},x_{-i}),\ \forall i.\] It should suffice then to calculate the value of just one of the $u_i(e_\alpha,x_{-i})$. We will use the $n$-linearity of the payoffs to complete the task. \begin{align} u_1(e_1,x_{-1}) & =  \sum_{j_2=1}^{s_2}\frac{1}{s_2}u_1(e_1,e_{j_2},x_3,\ldots, x_n) \\
& =  \frac{1}{s_2}\sum_{j_2=1}^{s_2}u_1(e_1,e_{j_2},x_3,\ldots, x_n) \\
& =  \frac{1}{s_2}\frac{1}{s_3}\sum_{j_2=1}^{s_2}\sum_{j_3=1}^{s_3}u_1(e_1,e_{j_2},e_{j_3},x_4,\ldots, x_n) \\
& = \ldots = \frac{1}{\prod_{i\in N/\{1\}}s_i}\sum_{j_2=1}^{s_2}\sum_{j_3=1}^{s_3}\cdots\sum_{j_n=1}^{s_n} u_1(e_1,e_{j_2},e_{j_3},\ldots ,e_{j_n}) \\
& = \frac{s_1}{|S|}\sum_\beta u_1(e_1,e_{-1\beta})
\end{align} which is exactly the average of the coefficients in the first row of $A_1$. Thus for any agent $i$ we have the equalities \[ \frac{s_i}{|S|}\sum_\beta u_i(e_1,e_{-i\beta}) = \frac{s_i}{|S|}\sum_\beta u_i(e_2,e_{-i\beta}) = \ldots = \frac{s_i}{|S|}\sum_\beta u_i(e_{s_i},e_{-i\beta})\] which after cancellation of the non-zero term $\frac{s_i}{|S|}$ proves our assertion.
\end{proof}
\begin{proposition}
If uniform distribution, $\hat{x}\in\Delta$, is a Nash equilibrium and in each of the payoff matrices the sums of the elements in each row that are larger than the average are equal, then it is an incentive equilibrium for $\varphi^D(x)$.
\end{proposition}
\begin{proof}
We will use the above lemma to prove the assertion. Given that the rows must all have an equal sum, the average of the elements in $A_i$, which we will denote $\bar{a}_i$, is equal to $\frac{s_i}{|S|}\sum_\beta a_{i1\beta}$. Let us now consider the condition for an incentive equilibrium when our incentive is given by $\varphi^D(x)$. At a Nash equilibrium we have the following calculation for each agent $i$
\begin{align}
\varphi^D_{i\alpha}(\hat{x}) & = \sum_\gamma(u_i(e_{i\alpha},\hat{x}_{-i})-u_i(\hat{x})+a_{\alpha\gamma}-u_i(e_{i\alpha},\hat{x}_{-i}))_+ \\
& = \sum_\gamma(a_{\alpha\gamma}-u_i(e_{i\alpha},\hat{x}_{-i}))_+ \\
& = \sum_\gamma\left(a_{\alpha\gamma}-\frac{s_i}{|S|}\sum_\beta u_i(e_{i\alpha},e_{-i\beta})\right)_+ \\
& = \sum_\gamma (a_{\alpha\gamma} - \bar{a}_i)_+
\end{align} where the second line is justified since $\hat{x}$ is a Nash equilibrium and thus $u_i(e_{i\alpha},\hat{x}_{-i})=u_i(\hat{x})$. The last line is simply the sum of all the elements from row $\alpha$ that are larger than the average. Given our assumption, it must be the case that $\varphi^D_{i\alpha}(\hat{x}) = \varphi^D_{i\beta}(\hat{x})$ for every $\alpha$ and $\beta$. Thus we have $\varphi^D_{i\alpha}(\hat{x}) = \frac{1}{s_i}\sum_\beta\varphi^D_{i\beta}(\hat{x})$ for every agent $i$, which is true if and only if $\hat{x}$ is an incentive equilibrium.
\end{proof}
To summarize, we found a class of games where the Nash equilibrium coincides with the incentive equilibrium for $\varphi^D(x)$ at the uniform distribution. All RPS games have the property that the rows of the payoff matrices are permutations of the first row. Games with this property form a subset of the games where the Nash equilibrium and our incentive equilibrium agree.

We conjecture that this is the only agreement outside of constant games and strategies where players are receiving their respective maximum payoff. There are simple counterexamples when either of the conditions is dropped. For example, if $A_i = \left(\begin{array}{rrr} 
1&0&0\\0&1&0\\0&-3&1
\end{array}\right)$, the average is 0, but the sums across rows are not equal. The interior Nash equilibrium is $\hat{x}=( (1/6,1/6,2/3), (1/6,1/6,2/3) )$ while the incentive equilibrium for $\varphi^D$ is the uniform distribution. On the other hand, if $A_i = \left(\begin{array}{rr} 
1&2\\3&0
\end{array}\right)$ then the Nash equilibrium is the uniform distribution, but the incentive equilibrium is $\hat{x} \approx ((0.31,0.69),(0.31,0.69))$.

\section{Asymptotic Stability}

As we have seen, many of the dynamics that have Nash equilibria as fixed points do not necessarily converge to the uniform distribution. The specific examples that do (at least so far) have been Rock-Paper-Scissors type games. We notice that the main idea is to create a cycle of best replies by permuting the values in the first row of the payoff matrix. This cyclic behavior is essentially the problem with convergence. We will now show that changing the parameters while maintaining this type of cyclic payoff structure has no impact on the asymptotic stability of the incentive equilibrium for $\varphi^D(x)$. 
\begin{proposition}
If the rows of the payoff matrix $A_i$ are permutations of each other, $\varphi^D_{i\alpha}(x)=\varphi^D_{i\beta}(x)$ for all $x\in\Delta$ and either $\varphi^D_{i\alpha}(\hat{x}) = 0$ or $\hat{x}_{i\alpha} = \frac{1}{s_i}$ for every $\alpha$ at incentive equilibrium.
\end{proposition}
\begin{proof}
Denote $\sigma$ as the permutation that takes row $\alpha$ to row $\beta$; then every element in row $\beta$ can be written as $a_{i\beta\gamma} = a_{i\alpha\sigma(k)}$ for some $k\in S_i$. Thus $\varphi^D_{i\beta} = \sum_\gamma (a_{i\beta\gamma}-u_i(x))_+ = \sum_k (a_{i\alpha\sigma(k)}-u_i(x))_+ = \varphi^D_{i\alpha}(x)$ regardless of $x$. 

We can now use this fact to describe all possible incentive equilibria for $\varphi^D(x)$. By definition, at equilibrium $\hat{x}$, $\varphi^D_{i\alpha}(\hat{x})=\hat{x}_{i\alpha}\sum_\beta\varphi^D_{i\beta}(\hat{x})$ for every $i$ and every $\alpha$. Given that the incentive functions are all equal regardless of $x\in\Delta$, we must have $\varphi^D_{i\alpha}(\hat{x})=\hat{x}_{i\alpha}s_i\varphi^D_{i\alpha}(\hat{x})$, which is true if and only if $\varphi^D_{i\alpha}(\hat{x})=0$ for all $\alpha$, which can occur only at the boundary or in a degenerate game, or when $\hat{x}_{i\alpha} = \frac{1}{s_i}$ for every $\alpha$.
\end{proof}
Recall the definition of an ISS is \[\hat{x}_i\cdot \frac{\varphi^D_i(x)}{x_i}> x_i \cdot \frac{\varphi^D_i(x)}{x_i}\] for all $x\neq\hat{x}$ in some neighborhood of $\hat{x}$. Also, an ISS is asymptotically stable wherever it satisfies the inequality in the definition. It will suffice then to prove that the uniform distribution is an ISS and the entire space is its basin of attraction.
\begin{theorem}
If $\hat{x}$ is a uniform incentive equilibrium for $\varphi^D(x)$ and the payoff matrices have rows that are permutations of each other, then $\hat{x}$ is globally asymptotically stable in $int\Delta$ for the incentive dynamics.
\end{theorem}
\begin{proof}
The previous proposition gives us that the incentives are equal for all $\alpha$ so we can without loss of generality use only $\varphi_{i1}^D(x)$ for each $i$.

\begin{align*}
0 & > - \frac{1}{s_i}\sum_\alpha \frac{\varphi^D_{i\alpha}(x)}{x_{i\alpha}} + \sum_\beta \varphi^D_{i\beta}(x) \\
& = -\frac{\varphi^D_{i1}(x)}{s_i}\sum_\alpha \frac{1}{x_{i\alpha}}+s_i\varphi^D_{i1}(x) \\
& = \frac{\varphi^D_{i1}(x)}{s_i}\left[s_i^2-\sum_\alpha \frac{1}{x_{i\alpha}}\right]\\
\end{align*}

If we define $f(x) = \sum_\alpha\frac{1}{x_{i\alpha}}$ it is easy to show that $f(x)$ has a global minimum of $s_i^2$ when $x_i$ is the uniform distribution. We simply optimize using Lagrange multipliers, noting that the Hessian matrix of $f(x)$ is positive definite in the interior of $\Delta$. Thus $\hat{x}$ satisfies the ISS definition for all $x\in int\Delta$.
\end{proof}

We further conjecture that all interior incentive equilibrium are asymptotically stable. If this is true, we can reduce the open problem of finding a game dynamic where every orbit converges to a rest point to proving that the basins of attraction for the incentive equilibrium form a partition of $\Delta$.

\bibliographystyle{amsalpha}
\bibliography{refs}

\end{document}